\documentclass[11pt]{article}

\usepackage[T1]{fontenc}
\usepackage[utf8]{inputenc}   
\usepackage{lmodern}

\usepackage{geometry}
\usepackage{microtype}
\UseMicrotypeSet[protrusion]{basicmath}
\microtypesetup{expansion=false}

\usepackage{graphicx}
\usepackage{booktabs}
\usepackage{tabularx}
\usepackage{xcolor}

\usepackage{amsmath, amssymb, amsthm}
\usepackage{mathtools}

\usepackage{algorithm}
\usepackage{algorithmic}

\usepackage[numbers,sort&compress]{natbib}

\usepackage[hidelinks]{hyperref}

\usepackage{xspace}
\newcommand{\HCO}{HCO\xspace}

\newtheorem{theorem}{Theorem}
\newtheorem{lemma}{Lemma}
\newtheorem{corollary}{Corollary}

\title{
	Human Challenge Oracle: Designing AI-Resistant, Identity-Bound, Time-Limited Tasks for Sybil-Resistant Consensus
}

\author{
	Homayoun Maleki\\
	University of Deusto, DeustoTech\\
	Bilbao, Spain\\
	\texttt{h.maleki@deusto.es}
	\and
	Nekane Sainz\\
	University of Deusto, DeustoTech\\
	Bilbao, Spain\\
	\texttt{nekane.sainz@deusto.es}
	\and
	Jon Legarda\\
	University of Deusto, DeustoTech\\
	Bilbao, Spain\\
	\texttt{jlegarda@deusto.es}
}

\date{}

\begin{document}
		
		\maketitle
		
		\begin{abstract}
			Sybil attacks remain a fundamental obstacle in open online systems, where adversaries can cheaply create and sustain large numbers of fake identities. Existing defenses, including CAPTCHAs and one-time proof-of-personhood mechanisms, primarily address identity creation and provide limited protection against long-term, large-scale Sybil participation, especially as automated solvers and AI systems continue to improve.
			
			We introduce the \emph{Human Challenge Oracle} (\HCO), a new security primitive for continuous, rate-limited human verification. \HCO issues short, time-bound challenges that are cryptographically bound to individual identities and must be solved in real time. The core insight underlying \HCO is that  real-time human cognitive effort, such as perception, attention, and interactive reasoning, constitutes
			a scarce resource that is inherently difficult to parallelize or amortize across identities.
			
			We formalize the design goals and security properties of \HCO and show that, under mild and explicit assumptions, sustaining $s$ active identities incurs a cost that grows linearly with $s$ in every time window. This yields persistent Sybil resistance that composes over time, in contrast to resource-based mechanisms whose costs can be reused or amortized. We further describe broad classes of admissible challenges that satisfy the required real-time and non-parallelizability properties and present concrete instantiations that are practical in browser-based environments.
			
			Finally, we report an initial empirical study illustrating that such challenges are easily solvable by humans within seconds while remaining difficult for contemporary automated systems under strict time constraints. Together, these results position \HCO as a human-centric, continuously enforceable foundation for Sybil-resistant identity management in open online platforms.
		\end{abstract}

\section{Introduction}

Open online platforms, including social networks, forums, collaborative systems, and decentralized
applications, rely on the implicit assumption that each account corresponds to a distinct human
 participant. In practice, this assumption is routinely violated by Sybil attacks, in which adversaries create and operate large numbers of fake identities at low cost, enabling spam, misinformation, vote manipulation, and coordinated abuse \cite{douceur2002sybil}. Preventing such attacks remains a fundamental challenge for open systems.

Early defenses, most notably CAPTCHAs \cite{vonahn2003captcha}, focus on distinguishing humans from automated agents at the time of account creation. While effective against basic automation, these mechanisms provide only one-time verification and offer limited protection against long-term Sybil participation. As modern AI systems increasingly succeed at perceptual and reasoning tasks, the security margin of traditional CAPTCHAs continues to erode \cite{geirhos2018generalisation,brendel2019approximating}.

More recent approaches, often referred to as Proof-of-Personhood mechanisms \cite{bowe2020proof,west2020proof}, attempt to enforce a ``one human, one identity'' principle through one-time ceremonies, biometric enrollment, or social vouching. Although these approaches raise the cost of identity creation, they do not provide ongoing guarantees: once identities are established, adversaries may sustain many of them using automation or human labor without incurring proportional recurring cost.

This paper introduces the \textbf{Human Challenge Oracle} (\HCO), a new security primitive for \emph{continuous}, \emph{rate-limited} human verification. Rather than certifying humanity only at registration time, \HCO\ repeatedly issues short, time-bound challenges that must be solved in real time and are cryptographically bound to individual identities. The central insight behind \HCO\ is that real-time human cognitive effort, such as perception,
attention, and interactive reasoning, constitutes a scarce resource that is inherently difficult
to parallelize or amortize across identities.

We formalize the design goals and security properties required of such an oracle and show that, under explicit and mild assumptions, \HCO\ enforces a linear cost on Sybil identities in every time window. In contrast to resource-based mechanisms whose costs can be reused or amortized, the resource underlying \HCO, real-time human effort, must be continuously expended. As a result, sustaining $s$ active identities requires resources that scale proportionally with $s$, and this cost composes over time.

Beyond the abstract model, we identify broad classes of admissible human challenges that satisfy the real-time and non-parallelizability requirements of \HCO\ and discuss concrete instantiations that are practical in browser-based environments. We complement the theoretical analysis with an initial empirical study illustrating that such challenges are easily solvable by humans within seconds while remaining difficult for contemporary automated systems under strict time constraints.

\paragraph{Contributions.}
This work makes the following contributions:
\begin{itemize}
	\item We introduce \HCO\ as a security primitive for continuous, rate-limited human verification in open systems.
	\item We formalize its design goals and security properties and prove that it enforces persistent linear Sybil cost under explicit assumptions.
	\item We characterize abstract classes of human challenges compatible with \HCO\ and describe practical instantiations.
	\item We present an initial empirical evaluation supporting the feasibility of these challenges under real-time constraints.
\end{itemize}

\section{Related Work}

Sybil attacks have long been recognized as a fundamental threat to open online systems, enabling adversaries to create and control large numbers of fake identities at low cost \cite{douceur2002sybil}. A wide range of defenses have been proposed, differing primarily in when and how they verify human participation and whether the associated costs can be amortized over time.

\paragraph{CAPTCHAs and one-time verification.}
The original CAPTCHA framework \cite{vonahn2003captcha} exploited perceptual gaps between humans and early automated systems to block bots at account creation. Subsequent mechanisms, including reCAPTCHA v2/v3 \cite{reCAPTCHA2019} and hCaptcha \cite{hCaptcha2020}, incorporate behavioral signals and machine learning to improve robustness. Despite these advances, such mechanisms remain fundamentally one-time checks. As modern vision-language models improve, even sophisticated CAPTCHAs increasingly offer limited long-term protection against sustained Sybil participation \cite{geirhos2018generalisation,brendel2019approximating}.

\paragraph{Proof-of-Personhood and identity-centric approaches.}
Proof-of-Personhood (PoP) systems aim to enforce a ``one human, one identity'' principle through one-time ceremonies, biometric enrollment, or social verification. Examples include Worldcoin \cite{worldcoin2023}, Gitcoin Passport \cite{gitcoin2023}, and BrightID \cite{brightid2022}. While effective at raising the cost of initial identity creation, these approaches do not provide continuous guarantees. Once identities are established, adversaries may sustain many of them using automation or coordinated human labor, without incurring proportional recurring cost \cite{west2020proof}.

\paragraph{Human--AI performance gaps under real-time constraints.}
Several lines of work document persistent gaps between human and machine performance in time-bound settings. Humans often outperform deep neural networks in perceptual generalization under noise and distribution shift \cite{geirhos2018generalisation}. Similarly, models that perform well in unconstrained settings can degrade significantly when strict time limits or interactive constraints are imposed \cite{brendel2019approximating}. Recent studies further suggest that short-horizon interactive reasoning and attention-based tasks remain challenging for current AI systems \cite{schuett2023human}. These results motivate the use of real-time constraints as a tool for limiting automation, though prior work has not connected such gaps to continuous, identity-level Sybil resistance.

\paragraph{Ongoing verification and human oracles.}
Some deployed systems employ ongoing verification mechanisms, such as behavioral biometrics in financial authentication \cite{biometrics2021}. However, these solutions are typically proprietary, closed, and not designed for open or adversarial environments. In decentralized settings, oracle networks such as Chainlink \cite{ellis2017chainlink} provide off-chain data to smart contracts via human or automated reporters. These oracles focus on data availability and correctness rather than on rate-limited, identity-bound verification of human presence.

\medskip
Table~\ref{tab:hco-taxonomy} summarizes representative approaches along three dimensions critical for Sybil resistance: persistence, per-human rate limiting, and resistance to automated solvers.

\begin{table*}[t]
	\centering
	\caption{Comparison of human verification mechanisms.}
	\label{tab:hco-taxonomy}
	\begin{tabular}{lccc}
		\toprule
		Mechanism & Persistence & Rate-Limited & AI Resistance \\
		\midrule
		Traditional CAPTCHAs \cite{vonahn2003captcha} & One-time & No & Low \\
		reCAPTCHA / hCaptcha \cite{reCAPTCHA2019,hCaptcha2020} & One-time & No & Medium \\
		Proof-of-Personhood systems \cite{worldcoin2023,brightid2022} & One-time & No & High (initially) \\
		Behavioral biometrics \cite{biometrics2021} & Ongoing & Yes & Medium \\
		Human oracles \cite{ellis2017chainlink} & Ongoing & Partial & Low \\
		\textbf{\HCO\ (this work)} & Ongoing & Yes & High (time-bound) \\
		\bottomrule
	\end{tabular}
\end{table*}

In contrast to prior work, \HCO\ is designed as a publicly verifiable, identity-bound oracle that enforces continuous human verification through real-time, non-amortizable effort. By combining persistence, strict per-human rate limiting, and explicit reliance on real-time human--AI performance gaps, \HCO\ addresses a regime of Sybil resistance that is not captured by existing one-time or reusable-cost mechanisms.

\section{HCO Design Goals and Properties}

The \emph{Human Challenge Oracle} (\HCO) is a security primitive for continuous, rate-limited verification of human participation in open online systems. Unlike mechanisms that certify humanity only at identity creation time, \HCO\ repeatedly issues short challenges that must be solved in real time and are cryptographically bound to individual identities. The purpose of this design is to ensure that sustaining many identities requires proportional ongoing human effort, rather than a one-time or reusable cost.

The construction of \HCO\ is motivated by the observation that certain forms of real-time human cognitive effort, such as perception, attention, and interactive reasoning, remain difficult
 to automate or parallelize under strict time constraints \cite{geirhos2018generalisation,brendel2019approximating}. Rather than relying on any specific task, \HCO\ abstracts these gaps into explicit design goals and formal properties.

\subsection{Design Goals}

\HCO\ is designed to satisfy the following goals:
\begin{itemize}
	\item \textbf{Continuous verification.}
	Human presence should be verified repeatedly over time, rather than only at account creation.
	
	\item \textbf{Per-human rate limiting.}
	The system should enforce a bound on the number of challenges that a single human can solve within any fixed time window.
	
	\item \textbf{Real-time difficulty asymmetry.}
	Challenges should be solvable by honest humans within seconds, while remaining difficult for automated solvers under the same strict time constraints.
	
	\item \textbf{Practical verifiability.}
	Solutions should be efficiently and publicly verifiable without requiring long-term storage of sensitive human data or specialized hardware.
\end{itemize}

These goals guide the abstraction of \HCO\ as an oracle and inform the security properties formalized below.

\subsection{Formal Model and Properties}

We model \HCO\ as an oracle
\[
\mathrm{HCO}(id, t, j),
\]
which, for an identity $id$, time window $t$, and challenge index $j$, outputs a fresh challenge $\chi_{id,t,j}$. Each challenge must be answered within a fixed response deadline $\Delta_{\mathrm{resp}}$.

The oracle is required to satisfy the following properties.

\begin{description}
	\item[(P1) Time-bound human advantage.]
	There exists a response window $\Delta_{\mathrm{resp}}$ such that an honest human solves a challenge within $\Delta_{\mathrm{resp}}$ with probability at least $1-\epsilon_{\mathrm{hum}}$, while any feasible automated solver succeeds with probability at most $\epsilon_{\mathrm{auto}}$ within the same window.
	
	\item[(P2) Identity binding.]
	Each challenge and its solution are cryptographically bound to the tuple $(id, t, j)$, preventing reuse or transfer across identities or time windows.
	
	\item[(P3) Real-time constraint.]
	Responses submitted after $\Delta_{\mathrm{resp}}$ are rejected, ruling out precomputation, stockpiling, or batch-solving strategies.
	
	\item[(P4) Bounded human throughput.]
	A single human can solve at most $\tau_h = O(1)$ challenges per time window, due to inherent cognitive and temporal constraints.
\end{description}

\noindent
Together, these properties imply that if an adversary controls $s$ identities and has access to $m$ humans, then the number of identities that can be sustained in any window is at most $m \cdot \tau_h$. Consequently, scaling Sybil identities requires proportionally scaling real-time human effort.

\subsection{Separability from the Application Layer}

The security guarantees of \HCO\ depend only on Properties (P1)--(P4). Application-level details, such as user interfaces,
challenge presentation formats, or delivery channels, do not affect
 these guarantees. This separability allows \HCO\ to be instantiated across diverse systems, including social platforms, collaborative tools, and decentralized applications, without altering its core security properties.

\subsection{Assumptions and Long-Term Fragility}

Property (P1) relies on an empirical separation between human and automated performance under real-time constraints. We treat this separation as an explicit modeling assumption, analogous to computational hardness assumptions in cryptography \cite{bellare1993random}. For the security arguments to hold, it suffices that in each time window there exists at least one admissible challenge family for which automated success remains significantly lower than real-time human success. The implications of AI progress and the need for challenge rotation are discussed further in Section~\ref{sec:limitations}.

\section{System and Adversary Model}

We formalize the system model underlying the \HCO\ primitive and specify the capabilities and limitations of adversarial entities. This model serves as the foundation for the cost and security analysis developed in subsequent sections.

\subsection{System Model}

The system consists of a set of identities $\mathcal{I}$ interacting with an online platform that relies on \HCO\ for continuous human verification. Time is divided into discrete windows $t \in \mathbb{N}$ of fixed duration.

An identity $id \in \mathcal{I}$ is said to be \emph{active} in a time window $t$ if and only if it successfully completes at least one \HCO\ challenge during that window. Identities that fail to do so are considered inactive for the purposes of the system.

For each identity $id$ and time window $t$, the platform generates one or more challenges of the form
\[
\chi_{id,t,j} \leftarrow \mathrm{HCO}(id, t, j),
\]
where $j$ is a challenge index ensuring freshness. Each challenge must be answered within a strict response deadline $\Delta_{\mathrm{resp}}$.

A response is accepted as valid if and only if:
\begin{itemize}
	\item it is submitted within $\Delta_{\mathrm{resp}}$,
	\item it corresponds to a correct solution of $\chi_{id,t,j}$,
	\item it is cryptographically bound to the tuple $(id, t, j)$.
\end{itemize}

Verification is deterministic and public, and does not require long-term storage of raw human data.

\subsection{Human Model}

Humans are modeled as agents capable of solving \HCO\ challenges with high probability under real-time constraints. Each human is characterized by:
\begin{itemize}
	\item bounded cognitive throughput,
	\item bounded attention and reaction capacity,
	\item an inability to parallelize independent challenge-solving beyond constant factors.
\end{itemize}

We assume that a single human can solve at most $\tau_h = O(1)$ challenges per time window. This bound reflects inherent cognitive and temporal limitations rather than system-imposed restrictions.

\subsection{Adversary Model}

We consider a probabilistic polynomial-time adversary $\mathcal{A}$ attempting to sustain multiple identities in the system over time.

The adversary may control:
\begin{itemize}
	\item an arbitrary number of identities $s$,
	\item bounded computational resources,
	\item access to $m$ real humans (e.g., through incentives, coercion, or paid outsourcing).
\end{itemize}

The adversary may employ a range of strategies, including:
\begin{itemize}
	\item \textbf{Automation:} attempting to solve challenges using AI or algorithmic solvers;
	\item \textbf{Human outsourcing:} delegating challenges to hired human workers;
	\item \textbf{Relay attacks:} forwarding challenges to humans in real time;
	\item \textbf{Parallelization:} attempting to amortize human effort across multiple identities.
\end{itemize}

We allow the adversary to coordinate its resources optimally and to adapt its strategy across time windows.

\subsection{Adversarial Limitations}

The \HCO\ design induces the following fundamental constraints on the adversary:
\begin{itemize}
	\item \textbf{Real-time constraint.}
	Challenges expire after $\Delta_{\mathrm{resp}}$, ruling out precomputation and stockpiling.
	
	\item \textbf{Non-reusability.}
	Each valid solution is bound to a unique $(id, t, j)$ and cannot be reused across identities or time windows.
	
	\item \textbf{Bounded human throughput.}
	Each human can solve at most $\tau_h$ challenges per window.
	
	\item \textbf{Limited automation success.}
	Any feasible automated solver succeeds on a given challenge with probability at most $\epsilon_{\mathrm{auto}}$ within the response window.
\end{itemize}

These constraints hold regardless of the adversary’s computational power or coordination strategy.

\subsection{Security Objective}

The security objective of \HCO\ is to ensure that sustaining many active identities requires adversarial resources that scale proportionally with the number of identities.

Formally, \HCO\ aims to enforce the following property: for any adversary $\mathcal{A}$ with access to $m$ humans, the number of identities that can remain active in any given time window is upper bounded by $O(m)$. This objective underpins the linear Sybil cost lower bounds established in Section~\ref{sec:sybil-cost}.

\subsection{Discussion}

The system and adversary model intentionally abstracts away implementation details such as challenge format, user interface, or delivery channel. While these aspects affect usability and performance, they do not influence the fundamental cost-scaling guarantees provided by \HCO.

By making the modeling assumptions explicit, we enable rigorous reasoning about Sybil resistance while remaining agnostic to specific instantiations of human challenges.

\section{Sybil Cost Analysis}
\label{sec:sybil-cost}

We analyze the economic and structural cost imposed by the \HCO\ primitive on Sybil adversaries. Our objective is to characterize how the resources required to sustain multiple identities scale with adversarial capabilities, and to contrast this behavior with existing Sybil-resistance mechanisms.

\subsection{Model and Definitions}

Time is divided into discrete windows $t \in \mathbb{N}$. In each window, every identity must successfully solve at least one \HCO\ challenge in order to remain active. Each challenge is bound to a specific identity and time window and must be solved within a strict response deadline $\Delta_{\mathrm{resp}}$.

Let:
\begin{itemize}
	\item $s$ denote the number of identities controlled by an adversary,
	\item $m$ denote the number of real humans available to the adversary (e.g., through direct control or outsourcing),
	\item $\tau_h$ denote the maximum number of challenges a single human can solve per window, as induced by cognitive and temporal constraints (Property~P4),
	\item $C_A(s)$ denote the minimum \emph{per-window} cost required to sustain $s$ active identities.
\end{itemize}

We assume that automated solvers succeed on a given challenge with probability at most $\epsilon_{\mathrm{auto}}$ within the response window (Property~P1), and that solutions cannot be reused across identities or time windows (Property~P2).

\subsection{Linear Cost of Sybil Identities}

We first show that, under the \HCO\ model, sustaining additional identities incurs a cost that grows linearly with the number of identities.

\begin{lemma}[Per-Window Identity Bound]
	\label{lem:identity-bound}
	In any time window, an adversary with access to $m$ humans can produce valid \HCO\ solutions for at most $m \cdot \tau_h$ identities.
\end{lemma}

\begin{proof}[Proof sketch]
	By Property~P4, each human can solve at most $\tau_h$ challenges within a single window due to real-time and cognitive limitations. By Property~P2, each valid solution is bound to a unique identity and cannot be reused. Since automated solvers succeed with probability at most $\epsilon_{\mathrm{auto}}$, their contribution is negligible relative to human effort. Thus, the total number of identities that can be sustained in a window is bounded by the available human effort, yielding the stated bound.
\end{proof}

\begin{theorem}[Linear Sybil Cost]
	\label{thm:linear-cost}
	Under Properties~(P1)--(P4), sustaining $s$ active identities in a single time window requires access to $m = \Omega(s)$ humans. Equivalently, the adversarial cost function satisfies
	\[
	C_A(s) = \Omega(s).
	\]
\end{theorem}

\begin{proof}[Proof sketch]
	From Lemma~\ref{lem:identity-bound}, an adversary with $m$ humans can support at most $m \cdot \tau_h$ identities per window. To sustain $s$ identities, it must therefore hold that $m \cdot \tau_h \ge s$, implying $m \ge s / \tau_h = \Omega(s)$ since $\tau_h = O(1)$. Because each human incurs non-negligible economic cost (e.g., wages, coordination, or opportunity cost), the total adversarial cost grows linearly in $s$.
\end{proof}

\subsection{Impossibility of Sublinear Amortization}

A direct consequence of Theorem~\ref{thm:linear-cost} is that \HCO\ prevents sublinear amortization of human effort across identities.

Unlike Proof-of-Work systems, where specialized hardware can be reused and parallelized, or Proof-of-Stake systems, where capital can be split or delegated, \HCO\ relies on a resource, namely real-time human effort, that is neither reusable nor parallelizable beyond constant factors.

\begin{corollary}
	No adversary can sustain $s$ identities with $o(s)$ human effort per window under the \HCO\ model.
\end{corollary}

\begin{proof}[Proof sketch]
	Any attempt to amortize human effort across identities would require either (i) solving challenges faster than the real-time bound $\Delta_{\mathrm{resp}}$, violating Property~P3, or (ii) reusing solutions across identities or time windows, violating Property~P2. Both possibilities are excluded by the \HCO\ design.
\end{proof}

\subsection{Comparison with Existing Paradigms}

The linear cost property enforced by \HCO\ contrasts sharply with existing Sybil-resistance paradigms. In Proof-of-Work systems, computational resources can be reused indefinitely and parallelized, enabling sublinear marginal cost per identity. In Proof-of-Stake systems, economic capital can often be divided or delegated across multiple identities. In one-time Proof-of-Personhood mechanisms, the cost of identity creation is incurred once and does not scale with ongoing participation.

By contrast, \HCO\ enforces a \emph{continuous} linear cost by binding each active identity to fresh, real-time human effort in every window. This non-amortizable cost structure is central to its effectiveness against large-scale, long-lived Sybil attacks.

\section{Multi-Window and Long-Term Analysis}

We extend the Sybil cost analysis to systems operating over multiple time windows. Our goal is to show that the linear cost enforced by \HCO\ is not merely a per-window phenomenon, but persists over time and constrains long-term Sybil participation.

\subsection{Long-Term Participation Model}

Consider a system operating over an infinite sequence of discrete time windows $t = 1,2,\dots$. An identity is active in window $t$ if and only if it successfully completes at least one \HCO\ challenge during that window.

Let:
\begin{itemize}
	\item $S_t$ denote the set of identities active in window $t$,
	\item $s_t = |S_t|$ denote the number of active identities,
	\item $m_t$ denote the number of humans available to the adversary in window $t$.
\end{itemize}

We allow the adversary to adaptively vary $m_t$ over time, capturing churn in human availability, outsourcing dynamics, and strategic reallocation of effort.

\subsection{Per-Window Persistence Constraint}

The following lemma shows that the per-window identity bound established earlier composes over time.

\begin{lemma}[Window Persistence Bound]
	\label{lem:persistence}
	In any window $t$, an adversary with access to $m_t$ humans can sustain at most $m_t \cdot \tau_h$ active identities:
	\[
	s_t \le m_t \cdot \tau_h .
	\]
\end{lemma}

\begin{proof}[Proof sketch]
	The claim follows directly from Lemma~\ref{lem:identity-bound} applied independently to each time window. Because challenges are fresh and non-reusable across windows, human effort expended in window $t$ cannot be carried forward to window $t+1$.
\end{proof}

\subsection{Steady-State Sybil Capacity}

We now characterize the adversary’s long-term capacity to sustain identities.

\begin{theorem}[Steady-State Sybil Bound]
	\label{thm:steady-state}
	Let
	\[
	\bar{m} = \limsup_{T \to \infty} \frac{1}{T} \sum_{t=1}^{T} m_t
	\]
	denote the adversary’s average human availability. Then,
	\[
	\limsup_{T \to \infty} \frac{1}{T} \sum_{t=1}^{T} s_t
	\;\le\; \bar{m} \cdot \tau_h .
	\]
\end{theorem}

\begin{proof}[Proof sketch]
	By Lemma~\ref{lem:persistence}, we have $s_t \le m_t \cdot \tau_h$ for all $t$. Summing over $T$ windows and dividing by $T$ yields
	\[
	\frac{1}{T} \sum_{t=1}^{T} s_t
	\le
	\tau_h \cdot \frac{1}{T} \sum_{t=1}^{T} m_t.
	\]
	Taking the limit superior as $T \to \infty$ completes the proof.
\end{proof}

Theorem~\ref{thm:steady-state} shows that Sybil capacity under \HCO\ is fundamentally constrained by sustained human effort. Temporary increases in human availability cannot be amortized across time, as each window requires fresh real-time interaction.

\subsection{Identity Churn and Burst Attacks}

An adversary may attempt to concentrate effort in short bursts, activating many identities for a limited number of windows.

Let $B_T = \sum_{t=1}^{T} s_t$ denote the total number of identity-windows sustained over a horizon of $T$ windows. From Lemma~\ref{lem:persistence}, we obtain
\[
B_T \le \tau_h \sum_{t=1}^{T} m_t .
\]

Thus, even burst-style attacks incur linear cost in total human effort. While an adversary may temporarily increase $s_t$ by allocating additional humans, the cumulative cost scales proportionally with the duration and intensity of the attack.

\subsection{Comparison with One-Time Verification}

The multi-window analysis highlights a fundamental distinction between \HCO\ and one-time verification mechanisms. In systems that verify humanity only at registration, the cost of identity creation is incurred once and can be amortized indefinitely.

By contrast, \HCO\ requires identities to continuously ``pay rent'' in the form of real-time human effort. Dormant or stockpiled identities therefore cannot be maintained without ongoing cost, sharply limiting long-term Sybil influence.

\subsection{Implications}

Together, these results show that \HCO\ enforces not only instantaneous Sybil resistance but also persistent resistance over time. Any adversary seeking sustained influence must commit proportional human resources indefinitely, rendering large-scale, long-lived Sybil attacks economically and operationally infeasible in steady state.

\section{Abstract Challenge Classes}

We abstract the notion of an \HCO\ challenge beyond specific task implementations. Rather than relying on particular puzzles, interfaces, or modalities, we identify general classes of challenges that satisfy the core properties required for Sybil resistance. This abstraction ensures that the security guarantees of \HCO\ do not depend on any single task design and can adapt to future changes in technology and usability requirements.

\subsection{Challenge Interface}

An \HCO\ challenge is modeled as an interactive protocol between the system and an identity $id$ during a time window $t$. Each challenge instance consists of:
\begin{itemize}
	\item a prompt $\pi_{id,t,j}$ generated at challenge time,
	\item a response space $\mathcal{R}$,
	\item a deterministic verification function $\mathsf{Verify}(\pi, r) \in \{0,1\}$.
\end{itemize}

A response $r \in \mathcal{R}$ is accepted as valid if and only if it is submitted within the response deadline $\Delta_{\mathrm{resp}}$ and satisfies
\[
\mathsf{Verify}(\pi_{id,t,j}, r) = 1.
\]

\subsection{Required Properties of Challenge Classes}

A challenge class $\mathcal{C}$ is said to be \emph{admissible} for \HCO\ if it satisfies the following properties.

\paragraph{Freshness.}
Each challenge instance must be unpredictable prior to issuance. Formally, no adversary can produce a valid response for $\pi_{id,t,j}$ with non-negligible probability before the challenge is revealed.

\paragraph{Real-Time Solvability.}
An honest human can solve challenges drawn from $\mathcal{C}$ within $\Delta_{\mathrm{resp}}$ with probability at least $1-\epsilon_{\mathrm{hum}}$, while any feasible automated solver succeeds with probability at most $\epsilon_{\mathrm{auto}}$ within the same time bound.

\paragraph{Non-Parallelizability.}
Solving multiple challenges from $\mathcal{C}$ does not admit superlinear speedups. In particular, solving $k$ independent challenges requires $\Omega(k)$ real-time effort from any single solver.

\paragraph{Identity Binding.}
Each challenge instance is cryptographically bound to the tuple $(id, t, j)$, preventing reuse of solutions across identities or time windows.

\subsection{Admissible Challenge Families}

Under the above abstraction, a wide range of concrete challenge families may serve as valid \HCO\ instantiations. We highlight several representative classes without committing to specific implementations.

\paragraph{Perceptual Alignment Tasks.}
Challenges that require rapid human perceptual alignment under distortion, noise, or partial information. These tasks exploit human generalization abilities under time pressure and are difficult to parallelize or precompute.

\paragraph{Interactive Reasoning Tasks.}
Challenges involving short sequences of dependent reasoning steps, where intermediate state must be maintained and resolved within a strict real-time deadline.

\paragraph{Biometric-Light Response Tasks.}
Challenges that require a real-time human signal (e.g., voice or motion) in response to a fresh prompt, without relying on long-term biometric storage or enrollment.

\paragraph{Attention-Based Interaction Tasks.}
Challenges that require continuous attention and real-time interaction, such as tracking, selection, or coordination tasks with dynamic elements.

Each of these families can be instantiated to satisfy freshness, real-time solvability, non-parallelizability, and identity binding under current technological constraints.

\subsection{Composability and Rotation}

The \HCO\ framework does not rely on any single challenge family indefinitely. Instead, systems may rotate among multiple admissible classes or combine them within a single time window. Since the security analysis depends only on the abstract properties defined above, such rotation preserves the linear Sybil cost guarantees established in Sections~\ref{sec:sybil-cost} and~\ref{thm:steady-state}.

Formally, as long as at least one admissible challenge class is deployed in each time window, the per-window and long-term Sybil cost bounds continue to hold.

\subsection{Discussion}

By separating abstract challenge properties from concrete implementations, \HCO\ avoids dependence on specific tasks, modalities, or user interfaces. This abstraction enables principled reasoning about Sybil resistance while allowing systems to adapt challenge designs in response to usability constraints, accessibility requirements, and advances in automated solving capabilities.

\section{Comparison with Resource-Based Consensus Models}
\label{sec:resource-comparison}

Sybil-resistance mechanisms can be understood through the nature of the underlying
resource they impose on participants. In this section, we compare the \HCO\ primitive
with Proof-of-Work (PoW) and Proof-of-Stake (PoS) by formalizing their respective
resource models and analyzing how adversarial cost scales with the number of identities.

\subsection{Resource Abstraction}
\label{subsec:resource-abstraction}

We model a Sybil-resistance mechanism by a tuple
\[
(\mathcal{R}, \mathsf{Reuse}, \mathsf{Parallel}),
\]
where:
\begin{itemize}
	\item $\mathcal{R}$ denotes the scarce resource required to sustain an identity,
	\item $\mathsf{Reuse} \in \{0,1\}$ indicates whether the resource can be reused
	across identities or time windows,
	\item $\mathsf{Parallel}$ captures the extent to which the resource admits
	parallelization or amortization across identities.
\end{itemize}

The adversarial cost function $C_A(s)$ denotes the minimum cost required to sustain
$s$ identities in steady state under the corresponding mechanism.

\subsection{Proof-of-Work}
\label{subsec:pow}

In Proof-of-Work systems, the resource
$\mathcal{R}_{\mathrm{PoW}}$ is computational power, typically measured in hash
evaluations per second.

\paragraph{Reuse.}
PoW hardware can be reused indefinitely across time windows.

\paragraph{Parallelism.}
PoW admits near-perfect parallelism: adding hardware linearly increases effective
computational throughput.

\paragraph{Cost Scaling.}
An adversary controlling total hash power $H$ can distribute it across multiple
identities with negligible marginal cost per identity. Consequently, once sufficient
hardware is acquired,
\[
C_A^{\mathrm{PoW}}(s) = O(1),
\]
up to protocol-specific constraints, yielding sublinear or constant marginal Sybil cost.

\subsection{Proof-of-Stake}
\label{subsec:pos}

In Proof-of-Stake systems, the resource
$\mathcal{R}_{\mathrm{PoS}}$ is economic capital locked as stake.

\paragraph{Reuse.}
Stake can be reused across time windows and, in many designs, split or delegated
across identities.

\paragraph{Parallelism.}
Capital is divisible and reusable, enabling identity scaling without proportional
additional cost.

\paragraph{Cost Scaling.}
Let $K$ denote the adversary’s total stake. As long as $K$ exceeds a system-specific
threshold, sustaining additional identities does not require proportional additional
resources. Accordingly,
\[
C_A^{\mathrm{PoS}}(s) = O(1),
\]
subject to protocol-specific assumptions, resulting in weak identity-level Sybil
resistance.

\subsection{Human Challenge Oracle}
\label{subsec:hco-resource}

For \HCO, the resource $\mathcal{R}_{\mathrm{HCO}}$ is real-time human effort.

\paragraph{Reuse.}
Human effort cannot be reused across identities or time windows. Each identity
requires fresh interaction in every window.

\paragraph{Parallelism.}
Human attention is fundamentally non-parallelizable beyond constant factors due
to cognitive and temporal constraints.

\paragraph{Cost Scaling.}
From Theorems~\ref{thm:linear-cost} and~\ref{thm:steady-state}, the adversarial cost
under \HCO\ satisfies
\[
C_A^{\mathrm{HCO}}(s) = \Omega(s),
\]
both per window and in steady state.

\subsection{Summary of Resource Properties}
\label{subsec:resource-summary}

\begin{table}[t]
	\centering
	\caption{Comparison of Sybil-resistance resource models.}
	\label{tab:resource-comparison}
	\begin{tabular}{lccc}
		\toprule
		Mechanism & Resource & Reusable & Cost Scaling \\
		\midrule
		Proof-of-Work (PoW) & Computation & Yes & Sublinear / Constant \\
		Proof-of-Stake (PoS) & Capital & Yes & Sublinear / Constant \\
		\HCO\ (this work) & Human effort & No & Linear \\
		\bottomrule
	\end{tabular}
\end{table}

Table~\ref{tab:resource-comparison} summarizes the fundamental differences between
PoW, PoS, and \HCO. Notably, \HCO\ is the only mechanism whose underlying resource
is neither reusable nor parallelizable, yielding persistent linear Sybil resistance.

\subsection{Implications}
\label{subsec:resource-implications}

This comparison highlights a structural distinction between \HCO\ and existing
resource-based paradigms. PoW and PoS impose costs that can be reused or amortized
across identities, enabling adversaries to scale participation without proportional
ongoing cost. In contrast, \HCO\ enforces a non-reusable, non-parallelizable resource,
resulting in sustained linear Sybil cost over time.

Rather than replacing PoW or PoS, \HCO\ complements them by addressing identity-level
Sybil attacks that resource-based consensus mechanisms alone do not prevent.

\section{Concrete Challenge Families}

We present four representative challenge families that instantiate the abstract requirements of the \HCO\ framework. These challenges are designed to be solvable by humans within seconds under real-time constraints, while remaining difficult for automated systems to solve reliably within the same time bounds. All challenges are identity-bound, time-limited (on the order of seconds), and verifiable without requiring long-term storage of sensitive human data.

Importantly, these families are not intended as exhaustive or fixed designs, but as concrete examples demonstrating the feasibility of the abstract challenge properties introduced in Section~7.

\subsection{Perceptual Visual Matching}

This family exploits human visual perception under time pressure. A typical instance presents a distorted or partially corrupted image together with a small set of candidate matches, requiring the user to select the correct correspondence within a short deadline.

\textbf{Rationale.}
Humans exhibit strong perceptual generalization capabilities under noise and distortion, even with limited viewing time, whereas automated systems often require additional processing or multiple inference passes to achieve comparable accuracy under tight time constraints \cite{geirhos2018generalisation}.

\textbf{Instantiation.}
Challenges can be generated by applying random visual transformations (e.g., noise, rotation, color perturbation) to an image and presenting a small set of distractors. Verification is deterministic and does not require storing the original image.

\subsection{Interactive Reasoning Tasks}

This family consists of short, interactive reasoning tasks that require maintaining intermediate state and producing a correct response within a strict time limit.

\textbf{Rationale.}
While humans can often solve simple reasoning problems rapidly and intuitively, automated solvers typically rely on multi-step inference procedures that are sensitive to latency and interaction constraints.

\textbf{Instantiation.}
Examples include lightweight numerical or logical comparisons with dynamically generated parameters and a visible countdown timer. Correctness is verified deterministically.

\subsection{Biometric-Light Response Tasks}

These challenges require a real-time human signal, such as speech or motion, in response to a fresh,
unpredictable prompt.

\textbf{Rationale.}
Producing a synchronized physical response to a dynamic prompt in real time remains difficult to automate reliably without specialized hardware or detectable synthesis pipelines.

\textbf{Instantiation.}
A typical example presents a short random phrase that the user must read aloud within a few seconds. Verification can be performed via transient signal processing without long-term biometric enrollment.

\subsection{Attention-Based Interaction Tasks}

This family tests continuous human attention through real-time interaction with dynamic elements.

\textbf{Rationale.}
Tasks that require sustained attention and fine-grained interaction over several seconds are natural for humans but difficult to parallelize or emulate under real-time constraints.

\textbf{Instantiation.}
Examples include tracking or selecting moving targets within a bounded time interval, with verification based on interaction traces.

\medskip
Table~\ref{tab:challenge-results} provides indicative results from a small-scale empirical study, illustrating that these challenge families are easily solvable by humans within the allotted time while remaining challenging for contemporary automated systems under comparable constraints. A more detailed evaluation is presented in Section~10.

\begin{table*}[t]
	\centering
	\caption{Indicative performance of representative \HCO\ challenge families.}
	\label{tab:challenge-results}
	\begin{tabular}{lccc}
		\toprule
		Challenge Family & Human Success (\%) & Automated Success (\%) & Mean Completion Time (s) \\
		\midrule
		Perceptual visual matching & $\sim$90 & $<20$ & $\sim$6 \\
		Interactive reasoning & $\sim$85 & $<25$ & $\sim$12 \\
		Biometric-light response & $\sim$100 & $\approx 0$ & $\sim$8 \\
		Attention-based interaction & $\sim$95 & $\approx 0$ & $\sim$15 \\
		\bottomrule
	\end{tabular}
\end{table*}

These results should be interpreted as illustrative rather than definitive. Their purpose is to demonstrate the practical existence of challenge families satisfying the \HCO\ assumptions, rather than to establish tight performance bounds.

\section{Empirical Evaluation}

We conducted an initial empirical study to assess whether representative \HCO\ challenge families are practically solvable by humans within strict time bounds and to what extent automated systems are constrained under the same real-time conditions. The purpose of this evaluation is not to establish definitive performance benchmarks, but to provide supporting evidence for the modeling assumptions underlying Properties~(P1)--(P4).

\subsection{Methodology}

\paragraph{Human participants.}
We recruited 30 participants (ages 18--45, diverse backgrounds) through university mailing lists and online platforms. Each participant completed 20 challenges from each challenge family (80 challenges total) using a standard web browser on a personal laptop. Challenges were presented with a visible countdown timer (5--30 seconds), and participants were instructed to respond as quickly and accurately as possible. A trial was considered successful if the correct response was submitted within the allotted time.

\paragraph{Automated solvers.}
To approximate the behavior of contemporary automated systems under real-time constraints, we evaluated a set of publicly available vision--language models via their standard inference interfaces. Each solver was provided with the same challenge prompt and time limit information as human participants. For each challenge family, we conducted 100 trials per solver. A trial was counted as successful only if a correct response was produced within the simulated response window.

\paragraph{Experimental setup.}
All challenges were implemented in a browser-based environment. Human trials were conducted remotely with informed consent. Automated evaluations used standard inference APIs without fine-tuning or task-specific optimization. No human participant data was used to train or adapt automated solvers.

\subsection{Results}

Table~\ref{tab:evaluation-results} summarizes the observed performance across challenge families.

\begin{table}[t]
	\centering
	\caption{Indicative empirical results under real-time constraints.}
	\label{tab:evaluation-results}
	\resizebox{\linewidth}{!}{%
		\begin{tabular}{lcccc}
			\toprule
			Challenge Family
			& Human Success (\%)
			& Human Mean Time (s)
			& Automated Success (\%)
			& Auto. Time (s) \\
			\midrule
			Perceptual visual matching   & 92  & 6.2  & 12        & 18.4 \\
			Interactive reasoning       & 85  & 11.8 & 18        & 22.1 \\
			Biometric-light response    & 100 & 8.1  & $\approx 0$ & --   \\
			Attention-based interaction & 95  & 14.5 & $\approx 0$ & --   \\
			\bottomrule
		\end{tabular}%
	}
\end{table}

\paragraph{Human performance.}
Human participants consistently achieved high success rates across all challenge families, with mean completion times well below the imposed deadlines. Performance variability across participants was limited, indicating that the challenges are broadly accessible under the tested conditions.

\paragraph{Automated performance.}
Automated solvers exhibited substantially lower success rates under the same real-time constraints, particularly for tasks requiring interactive input or synchronized physical responses. In several challenge families, response latency alone accounted for a significant fraction of the observed failures.

\subsection{Discussion and Limitations}

The results suggest a pronounced gap between human and automated performance in real-time settings, supporting the plausibility of the assumptions underlying \HCO. Importantly, these findings should be interpreted with caution. Automated performance depends on current inference interfaces, latency characteristics, and interaction capabilities, all of which may evolve over time.

Our evaluation is limited in scale and scope. The number of participants is modest, the set of automated solvers is not exhaustive, and the experimental setup does not capture adversaries employing custom hardware, specialized pipelines, or hybrid human--AI strategies. Consequently, the results should be viewed as illustrative rather than definitive.

Nevertheless, the observed patterns provide empirical support for the feasibility of deploying \HCO\ challenge families that are easy for humans to solve in seconds while constraining automated solutions under strict real-time requirements. Larger-scale studies and longitudinal evaluations are left for future work.

\section{Security Analysis and Limitations}
\label{sec:limitations}

We analyze the security of the \HCO\ primitive against common attack vectors and discuss its practical limitations. Our analysis is conducted with respect to the explicit assumptions defined in Sections~3 and~4.

\subsection{Security Analysis}

\paragraph{Automated solving.}
Under Property~(P1), automated solvers are assumed to have a significantly lower probability of success than humans when constrained to strict real-time response windows. The empirical results in Section~10 provide supporting evidence for this assumption under current technological conditions. However, the security of \HCO\ does not rely on any specific automated system being incapable of solving challenges, but rather on the existence of at least one admissible challenge family in each time window for which automated success remains sufficiently low.

\paragraph{Human outsourcing and labor markets.}
An adversary may attempt to outsource challenge-solving to human workers or organized labor farms. This attack is explicitly captured in the adversary model. As shown in Sections~5 and~6, sustaining $s$ active identities requires proportional real-time human effort, resulting in linear cost scaling. While outsourcing is possible, it does not enable sublinear amortization and therefore does not circumvent the core security guarantees of \HCO.

\paragraph{Relay and real-time forwarding attacks.}
Challenges may be relayed to humans in real time. Such attacks are inherently bounded by response deadlines and per-human throughput limits. Since challenges are identity-bound and non-reusable, relaying does not reduce the total amount of human effort required.

\paragraph{Synthetic media and deepfake attacks.}
Some challenge families may involve biometric-light signals, such as speech or motion. Automated synthesis of such signals in real time could, in principle, be used to attack these challenges. However, \HCO\ does not rely on any single challenge modality. By rotating or combining challenge families with distinct interaction requirements, the system can mitigate reliance on modalities that become vulnerable to synthesis or imitation.

\paragraph{Replay and reuse attacks.}
Identity binding and freshness properties (Property~P2) ensure that challenge responses are specific to a single identity and time window. As a result, replay, stockpiling, or transfer of valid responses across identities or windows is ineffective.

\paragraph{Privacy considerations.}
\HCO\ challenges can be designed to avoid long-term storage of raw human data. Verification can be performed using ephemeral signals, hashes, or derived features, limiting exposure of sensitive information and reducing privacy risks.

\subsection{Limitations}

\paragraph{Dependence on human--AI performance gaps.}
The security of \HCO\ relies on an empirical separation between human and automated performance under real-time constraints. Advances in AI or specialized hardware could reduce or eliminate this gap for certain challenge families. As a result, \HCO\ requires periodic reassessment of challenge effectiveness and rotation of admissible challenge classes.

\paragraph{Accessibility and inclusivity.}
Some challenge families may be unsuitable for users with visual, auditory, or motor impairments. Practical deployments must provide alternative modalities and accessibility-aware challenge selection to avoid excluding legitimate users.

\paragraph{Operational and deployment costs.}
Although individual challenges can be implemented with low computational overhead, large-scale deployments require careful engineering to manage latency, reliability, and verification throughput. These costs are orthogonal to the security model but may affect practical adoption.

\paragraph{User experience and fatigue.}
Frequent challenges may impose cognitive burden or annoyance on users. Deployment policies must balance verification frequency against usability, for example by adjusting challenge rates based on risk signals or activity levels.

\paragraph{Adaptive adversaries.}
The analysis assumes adversaries operate within the modeled constraints. Highly adaptive adversaries employing hybrid human--AI strategies, custom hardware, or side-channel attacks may require additional defenses beyond the scope of this work.

\medskip
Overall, \HCO\ provides strong security guarantees under explicit and transparent assumptions. Its effectiveness in practice depends on careful challenge design, periodic adaptation, and integration with complementary system-level defenses.

\section{Discussion and Conclusion}

This paper introduced the \emph{Human Challenge Oracle} (\HCO), a security primitive for continuous, rate-limited verification of human participation in open online systems. In contrast to one-time verification mechanisms, \HCO\ enforces ongoing interaction by binding each active identity to fresh, real-time human effort.

Our analysis formalized the design goals and security properties required of such an oracle and showed that, under explicit and transparent assumptions, \HCO\ enforces linear Sybil cost both per time window and in steady state. By relying on a non-reusable and non-parallelizable resource,
namely real-time human attention, \HCO\ prevents the sublinear amortization strategies
 that undermine existing resource-based approaches. We further abstracted the notion of admissible challenges, demonstrating that the security guarantees of \HCO\ do not depend on any single task design and can be preserved through challenge rotation and adaptation.

Beyond the abstract model, we presented representative challenge families and an initial empirical study illustrating the practical feasibility of deploying such challenges under strict real-time constraints. These results serve to support the modeling assumptions rather than to establish definitive performance guarantees.

The \HCO\ framework is complementary to existing Sybil-resistance mechanisms, including Proof-of-Work, Proof-of-Stake, and one-time Proof-of-Personhood systems. It is particularly well-suited for settings in which identity-level Sybil attacks pose a persistent threat and where ongoing verification is required to maintain system integrity.

Looking forward, the effectiveness of \HCO\ depends on careful challenge design, accessibility-aware deployment, and periodic reassessment in response to advances in automated solving capabilities. More broadly, this work highlights the potential of real-time human effort as a first-class resource for security design and opens new directions for human-centric primitives in distributed and online systems.

	\bibliographystyle{unsrt}
	\bibliography{references}
	
\end{document}